%% file: main.tex
\documentclass[11pt]{article}

\input{pre_doc}

\usepackage[utf8]{inputenc} 
\usepackage[T1]{fontenc}    
\usepackage{hyperref}       
\usepackage{url}            
\usepackage{booktabs}       
\usepackage{amsfonts}       
\usepackage{nicefrac}       
\usepackage{microtype}      
\usepackage{xcolor}         
\usepackage{algorithm}
\usepackage{algpseudocode}
\usepackage[numbers]{natbib}
\usepackage[margin=1in]{geometry}

\renewcommand{\narxiv}[1]{}

\title{Privately Clipping Heavy-Tailed Data}
\author{Matthew Joseph\thanks{mtjoseph@google.com. Google Research.} \and Alex Kulesza\thanks{kulesza@google.com. Google Research.}}

\begin{document}

\maketitle

\begin{abstract}
  Summing vectors is a basic task in differentially private data analysis. Standard private sum algorithms clip vectors to some bound, sum them, and add noise scaled to the bound. The differential privacy literature has developed a deep library of noise addition mechanisms, but it offers few tools for privately choosing a clipping bound based on the data. We introduce a method that privately chooses a clipping bound by optimizing \emph{residual coherence}, the directional alignment of vectors exceeding the bound, against the variance cost of additional noise. We prove a close relationship between residual coherence and bias for a general class of heavy-tailed data regimes and show empirically that, across several datasets, our method outperforms existing baselines.
\end{abstract}

\input{sections/intro}
\input{sections/prelims}
\input{sections/algorithm}
\input{sections/experiments}
\input{sections/discussion}
\input{sections/ai}

\newpage

\bibliographystyle{plainnat}
\bibliography{references}
\newpage
\appendix
\input{sections/appendix}

\end{document}

%% file: pre_doc.tex
\usepackage{amsmath}
\usepackage{amssymb}
\usepackage{amsthm}
\usepackage{amsfonts}
\usepackage{mathtools}
\usepackage{bbm}
\usepackage{pifont}
\usepackage{dsfont}
\usepackage{microtype}
\usepackage{booktabs}
\usepackage{subcaption}
\usepackage{graphicx}
\usepackage[dvipsnames]{xcolor}
\usepackage{algorithm}

\usepackage{thmtools}
\usepackage{thm-restate}

\usepackage[colorlinks, citecolor=blue, linkcolor=blue, urlcolor=blue]{hyperref}

\usepackage[capitalize,noabbrev]{cleveref}
\creflabelformat{equation}{#2#1#3}

\input{macros}

%% file: macros.tex
\newtheorem{theorem}{Theorem}
\newtheorem{definition}{Definition}
\newtheorem{lemma}{Lemma}
\newtheorem{assumption}{Assumption}

\newcommand{\abalg}{\mathsf{ApproxBounds}}
\newcommand{\atalg}{\mathsf{AboveThreshold}}
\newcommand{\abovethreshold}[1]{\mathsf{AboveThreshold}\left(#1\right)}
\newcommand{\bdkualg}{\mathsf{CoinPress}}

\newcommand{\calD}{\ensuremath{\mathcal{D}}}
\newcommand{\calO}{\ensuremath{\mathcal{O}}}

\newcommand{\calX}{\ensuremath{\mathcal{X}}}

\newcommand{\clip}[2]{\mathsf{clip}\left(#1, #2\right)}
\newcommand{\cov}[1]{\mathsf{Cov}\left(#1\right)}
\newcommand{\E}[2]{\mathbb{E}_{#1}\left[ #2 \right]}
\newcommand{\expo}[1]{\mathsf{Expo}\left(#1\right)}

\newcommand{\hlyalg}{\mathsf{HLY\text{-}UQ}}
\newcommand{\ksualg}{\mathsf{CDPHDME}}

\newcommand{\mse}[1]{\mathrm{MSE}\left(#1\right)}
\newcommand{\optalg}{\mathsf{Optimal}}
\renewcommand{\P}[2]{\mathbb{P}_{#1}\left[#2\right]}
\newcommand{\rc}[1]{\mathsf{RC}\left(#1\right)}
\newcommand{\rcalg}{\mathsf{OptimizeRC}}
\newcommand{\rhoclip}{\rho_{\textsc{clip}}}
\newcommand{\rhosum}{\rho_{\textsc{sum}}}
\newcommand{\tr}[1]{\mathsf{Tr}\left(#1\right)}
\newcommand{\uqalg}{\mathsf{UnboundedQuantile}}
\newcommand{\uqalgshort}{\mathsf{UQ}}

\newcommand{\narxiv}[1]{#1}

%% file: sections/intro.tex
\section{Introduction}
A differentially private (DP)~\cite{DMNS06} computation controls any single data contributor's impact on the final result. For DP vector sums, this control is typically achieved by clipping vectors to a maximum norm $C$, summing the clipped vectors, and adding noise scaled to $C$. Larger clipping bounds reduce bias but increase noise, so $C$ must be chosen carefully to minimize overall error.

The DP literature offers little guidance for this choice (see Section~\ref{subsec:related_work}). Instead, it typically assumes that the algorithm operator has enough domain knowledge to choose a good clipping bound without examining the dataset. This is sometimes plausible. Most people, for example, can probably provide tight bounds on human age or height. When prior knowledge is lacking, a slightly more involved solution is to clip to a private estimate of an extreme (e.g., 0.99) quantile of the data norms. In general, choosing a clipping bound should be straightforward for any data \emph{light-tailed} enough that typical points and outliers lie on similar scales.

This is not true for \emph{heavy-tailed} data, where a significant minority of users may contribute vectors orders of magnitude larger than those of the ``typical'' core. A large body of empirical work has observed that heavy-tailed distributions arise frequently across a variety of real-world contexts~\cite{Z49, G99, P14}, and are particularly common in internet activity~\cite{AH02}. Unfortunately, light-tailed heuristics fail in this setting. The 0.9, 0.95, and 0.99 norm quantiles likely lie on vastly different scales, and there is no principled reason to favor one in general: the 0.99 quantile may work well when a dataset's large vectors are well-aligned and poorly when they are not.

Perhaps not coincidentally, empirical work finds that selecting a clipping bound is a significant challenge to deploying DP~\cite{SSHSV23, CDEGH+23, AKV24}. This persists even though private sums (often in the form of private histograms) are a common primitive in public applications of DP~\cite{A23, BDDDE+20, WZLDS+20, G26, ABDHJ+23}.

\paragraph{Contributions.} This work proposes a fully differentially private (and, we suggest, effectively parameter-free) method for choosing an $\ell_2$ clipping bound $C$ for computing a private $d$-dimensional sum. The method applies $\atalg$~\cite{DNRRV09}, a canonical DP algorithm for privately answering a sequence of low-sensitivity binary queries with privacy cost scaling only with the number of ``yes'' responses. We use $\atalg$ to take a candidate sequence of clipping bounds $C_1 < C_2 < \cdots$ and check the current bound's \emph{residual coherence} against a fixed threshold.

\begin{definition}
\label{def:rc}
    Given a database of vectors $X = \{x_1, \ldots, x_n \in \mathbb{R}^d\}$ with norms $R_i = \|x_i\|_2$ and unit directions $U_i = x_i / R_i$, along with candidate clipping bound schedule $0 < C_1 < C_2 < \cdots$, the \emph{residual coherence} at $C_j$ is
    \begin{equation*}
        \rc{C_j; X} \coloneqq \left\|\sum_{R_i > C_j} U_i\right\|_2.
    \end{equation*}
\end{definition}
Importantly, residual coherence has bounded sensitivity of 2 under swap DP and 1 under add-remove DP. If the vectors are nonnegative, both sensitivities are 1. This makes it amenable to an algorithm like $\atalg$.

$\rc{C; X}$ (which we will shorthand as $\rc{C}$) is large exactly when the database has many well-aligned vectors whose norms exceed $C$. A simple motivation is that heavy-tailed data often exhibits a two-phase trend: there is a large ``core'' of individuals with similar small vectors (e.g., casual movie viewers selecting a few popular movies) and a smaller ``tail'' of individuals diverging into more numerous niche interests. Increasing the clipping bound to shrink residual coherence can therefore approximate the point at which the signal of the clipped tail is outweighed by the noise it demands.

A full description of our algorithm, $\rcalg$, appears in Section~\ref{subsec:alg_desc}, and we describe an efficient implementation in Section~\ref{subsec:rc_imp}. Section~\ref{subsec:rc_proxy} develops a simple utility result showing that residual coherence serves as a principled proxy for bias reduction in (a stylized class of) heavy-tailed data. Section~\ref{sec:experiments} evaluates our method across over a dozen real-world datasets.

\subsection{Related Work}
\label{subsec:related_work}
\citet{AKMV19} show that the optimal clipping bound for an $\varepsilon$-DP scalar sum is the $1 - \frac{1}{\varepsilon}$ quantile, but do not consider vector sums.~\citet{D23} suggests using $\atalg$ to estimate a fixed quantile of the data norms as a clipping bound and calls the method $\uqalg$ (which we shorthand $\uqalgshort$) as it does not require \emph{a priori} bounds. While this avoids selecting $C$ directly, it does so by shifting the hyperparameter burden to the choice of quantile. Furthermore, as we demonstrate empirically, the optimal quantile varies significantly across heavy-tailed datasets (Section~\ref{sec:experiments}) and no fixed quantile yields a good bias-variance tradeoff in general. Methods that simply attempt to privately estimate the maximum of a dataset~\cite{WZLDS+20, NA25} face similar issues.

A private sum can also be computed, in a somewhat roundabout way, by privately estimating the mean and scaling its output by the database size $n$. Most mean estimation algorithms in the literature are ruled out by focusing on light-tailed (e.g., subgaussian) distributions or requiring their own careful parameter selections. One exception is the quantile-based algorithm of~\citet{HLY21}, which we discuss and evaluate in Section~\ref{sec:experiments}.

Several works have considered clipping gradients for DP-SGD. AdaCliP~\cite{PSYRK19} and GeoClip~\cite{GTSK25} privately track gradient mean and variance statistics and use them to update subsequent clipping norm choices, but this consumes privacy budget over the course of many updates and introduces new hyperparameters that are tuned non-privately. \citet{ATMR21} propose an adaptive strategy that privately tracks a specified quantile (e.g., the median) of the update norm distribution during training, which shares the same general quantile drawbacks described above.~\citet{BWZK23} introduce a method that avoids choosing a clipping norm by normalizing each gradient to a unit vector, but this again shifts the choice to another hyperparameter (the learning rate), which is tuned non-privately.

%% file: sections/prelims.tex
\section{Preliminaries}
We use the ``swap'' notion of neighboring databases, primarily for consistency with the setup used by~\citet{D23}. Our algorithm is easily modified for add-remove privacy.
\begin{definition}
    Two databases $X, X'$ from domain $\calX$ are \emph{neighbors}, denoted $X \sim X'$, if one can be obtained from the other by replacing at most one data point.
\end{definition}
We use concentrated differential privacy (CDP) for similar reasons. The basic ideas of our method readily extend to other notions like pure and approximate DP, at some cost to composition neatness.
\begin{definition}[\cite{BS16}]
    Randomized mechanism $M \colon \calX \to \calO$ is $\rho$-CDP if for all $X \sim X' \in \calX$ and all $\alpha > 1$, $D_\alpha(M(X) || M(X')) \leq \rho \alpha$, where $D_\alpha$ denotes $\alpha$-Rényi divergence.
\end{definition}

Our method is tuned assuming that the final private sum is computed by adding $\rho$-CDP (spherical) Gaussian noise. The derivation extends to any noise with a closed-form expression for the trace of its covariance matrix, but we use $\rho$-CDP Gaussian noise for its generality, widespread use, and clean privacy composition. We recall the basic statement for $\rho$-CDP Gaussian noise below. 
\begin{lemma}[\cite{BS16}]
    Given statistic $f \colon \calX \to \mathbb{R}^d$ of $\ell_2$ sensitivity $C$ and database $X$, releasing $f(X) + Z$ where each $Z_j \sim_{iid} N(0, \tfrac{C^2}{2\rho})$ is $\rho$-CDP.
\end{lemma}

Although $\atalg$ was introduced by~\citet{DNRRV09}, we use the formulation (and improved analysis) of~\citet{D23}. Pseudocode appears in Algorithm~\ref{alg:at}, where $\expo{\lambda}$ has density $f_\lambda(x) \propto \tfrac{1}{\lambda}\exp(-x/\lambda)$ and support $\mathbb{R}_{\geq 0}$.

\begin{algorithm}[h]
\caption{$\atalg$}
\label{alg:at}
\begin{algorithmic}[1]
\Require Dataset $X$, queries $\{f_j : \mathcal{X} \rightarrow \mathbb{R}\}$ with sensitivity $\Delta$, threshold $T$, privacy budget $\rho$
\State Set $\tilde T = T + \expo{\Delta / \sqrt{2\rho}}$
\For{each query $j$}
        \State Set $\nu_j = \expo{2\Delta / \sqrt{2\rho}}$
        \State \textbf{if} $f_j(X) + \nu_j \ge \tilde{T}$ \textbf{then return} $\top$ and halt \textbf{else return} $\bot$ and continue
\EndFor
\end{algorithmic}
\end{algorithm}

\begin{lemma}[Corollary 3.1 and Theorem 1~\cite{D23}]
\label{lem:at_cdp}
    $\atalg$ is $4\rho$-CDP.
\end{lemma}

%% file: sections/algorithm.tex
\section{Optimizing Residual Coherence}
\subsection{Algorithm Description}
\label{subsec:alg_desc}
The starting point for our algorithm is a straightforward decomposition of the final error of a DP sum. As mentioned earlier, for simplicity we assume that the private sum is computed by clipping each vector to $\ell_2$ norm $C$, summing the clipped vectors, and then adding $\rhosum$-CDP Gaussian noise tuned to $\ell_2$ sensitivity $2C$, though our method can handle any form of spherical noise with a closed expression for its variance.

We measure error using Mean Squared Error (MSE) of the final vector sum. Given database of vectors $\{x_1, \ldots, x_n \in \mathbb{R}^d\}$, define norms $R_i \coloneqq \|x_i\|_2$ and unit directions $U_i \coloneqq \frac{x_i}{\|x_i\|_2}$. For initial exposition, suppose additionally that the vectors are nonnegative (see Section~\ref{subsec:rc_proxy} for analysis of the general case). A standard bias-variance decomposition yields
\begin{align*}
    \mse{C} =&\ \frac{1}{d}\E{}{\left\|\sum_{i=1}^n x_i - \left(\sum_{i=1}^n \clip{x_i}{C} + Z\right)\right\|_2^2} \\
    =&\ \frac{1}{d}\left\|\E{}{\sum_{i=1}^n x_i - \clip{x_i}{C}}\right\|_2^2 + \frac{\tr{\cov{Z}}}{d} \\
    =&\ \frac{1}{d}\left\|\sum_{R_i > C} (R_i-C)U_i\right\|_2^2 + \frac{2C^2}{\rhosum}
\end{align*}
since the sum clipped to $C$ has swap sensitivity $2C$. To find the optimal bound $C^*$, we take the derivative of the MSE with respect to $C$\footnote{The exact MSE is technically only piecewise differentiable due to the clipping operation, but evaluating its derivative where it exists provides the central intuition.}:
\begin{align*}
    \frac{d}{dC} \mse{C} =&\ \frac{2}{d}\sum_{R_i > C} (R_i-C)U_i \cdot \left(-\sum_{R_i > C} U_i\right) + \frac{4C}{\rhosum} \\
    =&\ -\frac{2}{d}\sum_{R_i > C}\sum_{R_j > C}(R_i-C)\langle U_i, U_j \rangle + \frac{4C}{\rhosum}.
\end{align*}
When the underlying vectors are nonnegative, so is each $\langle U_i, U_j \rangle$, and the overall MSE derivative goes from negative to positive as $C$ grows. We therefore set the derivative expression to 0 and get
\begin{equation}
\label{eq:optimal_C}
     \sum_{R_i > C}\sum_{R_j > C}(R_i-C)\langle U_i, U_j \rangle = \frac{2Cd}{\rhosum}.
\end{equation}
The LHS sum in Equation~\ref{eq:optimal_C} is the bias derivative. Evaluating it requires knowing the private magnitude depths $R_i - C$. Since these have unbounded sensitivity without clipping, it is not obvious how to estimate them in a differentially private way.  We instead approximate the bias derivative by using residual coherence (Definition~\ref{def:rc}) as a proxy:
\begin{equation}
\label{eq:approximation}
    \sum_{R_i > C}\sum_{R_j > C}(R_i-C)\langle U_i, U_j \rangle \approx C \cdot \rc{C}^2.
\end{equation}
Unlike the bias derivative, $\rc{C}$ has bounded sensitivity; Section~\ref{subsec:rc_proxy} describes other properties that justify its use. Our algorithm uses $\atalg$ to search for the threshold where $C \cdot \rc{C}^2$ balances the variance derivative $\frac{2Cd}{\rhosum}$, i.e., residual coherence equals $\sqrt{\frac{2d}{\rhosum}}$.  Pseudocode appears in Algorithm~\ref{alg:rc}. Note that we negate the queries and threshold because $\atalg$ stops when the query exceeds the threshold.

\begin{algorithm}[h]
\caption{$\rcalg$}
\label{alg:rc}
\begin{algorithmic}[1]
    \Require Dataset $X$, ambient dimension $d$, bounds $\{C_j\}$, budgets $\rhosum, \rhoclip$
    \State $T \leftarrow -\sqrt{\frac{2d}{\rhosum}}$ 
    \State Return $C_j$ where $\abovethreshold{X, -\{\rc{C_j}\}_{j=1}^\infty, \Delta = 1, T, \rhoclip / 4}$ halts
\end{algorithmic}
\end{algorithm}

To see that $\rcalg$ is private, observe that $\rc{C}$ is 1-sensitive for nonnegative vectors and 2-sensitive for general vectors, then apply Lemma \ref{lem:at_cdp}. 
\begin{lemma}
\label{lem:rc_cdp}
    If the data vectors are nonnegative, then $\rcalg$ is $\rhoclip$-CDP. For general vectors, it is $4\rhoclip$-CDP.
\end{lemma}

\subsection{Implementing $\rcalg$}
\label{subsec:rc_imp}
We briefly describe an efficient implementation of $\rcalg$ that uses pre-computation to speed up on-the-fly calculation of $\rc{C_j}$. We focus on the case where the candidate bounds $C_j$ are powers of some base $\beta$, as used by $\uqalgshort$~\cite{D23} and all of our experiments (see Section~\ref{sec:experiments} for additional discussion).

Pseudocode appears in Algorithm~\ref{alg:fast_rc}. The algorithm makes a first pass through the data to compute each data point $x_i$'s norm, unit vector $U_i$,  and the index $b_i$ of the largest clip bound $C_{b_i}$ where it contributes to residual coherence. It then computes the last round $B = \max\{b_i\}$ with nonzero residual coherence and constructs a $(B+1) \times d$ matrix $\texttt{dirs}$ where row $j$ stores directions for points $x_i$ with $b_i = j$. If $j > B$, then $\rc{C_j} = 0$. For the remaining clip bounds, we iterate backward to compute $\rc{C_B}, \rc{C_{B-1}}, \ldots$ using a running sum $\texttt{Tail}$ of each entry of $\texttt{dirs}$.

The result is a complete set of the residual coherence values needed for $\atalg$ in time and space $O(nd + Bd)$, with additional possible savings using sparse matrices and parallelism.  Empirical runtime results for our implementation, which uses Algorithm~\ref{alg:fast_rc} and Scipy's sparse matrix libraries, appear in Section~\ref{subsubsec:time}.

\begin{algorithm}[h]
\caption{Efficient $\rcalg$ Pre-Computation}
\label{alg:fast_rc}
\begin{algorithmic}[1]
\Require Dataset $X \in \mathbb{R}^{n \times d}$, geometric base $\beta$
\For{nonzero data points $x_i$}
    \State $R_i \gets \|x_i\|_2$
    \State $U_i \leftarrow x_i / R_i$
    \If{$R_i \leq 1$}
        \State $b_i \leftarrow -1$ \Comment{-1 handles points which contribute to no $\rc{\beta^j}$}
    \Else
        \State $b_i \gets \lfloor \log_\beta(R_i) \rfloor$
    \EndIf
\EndFor
\State $B \leftarrow \max_{i \in [n]} b_i$
\State Initialize array \texttt{dirs} of size $(B + 1) \times d$ to 0
\For{$i \in [n]$}
    \If{$b_i \geq 0$}
        \State $\text{\texttt{dirs}}[b_i] \leftarrow \text{\texttt{dirs}}[b_i] + U_i$
    \EndIf
\EndFor
\State $\texttt{Tail} \leftarrow \mathbf{0} \in \mathbb{R}^d$
\State Initialize array $\texttt{RC}$ of size $B+1$
\For{$k = B, B-1, \dots, 0$}
    \State $\texttt{Tail} \leftarrow \texttt{Tail} + \texttt{dirs}[k]$
    \State $\texttt{RC}[k] \leftarrow \|\texttt{Tail}\|_2$
\EndFor
\State \Return $\texttt{RC}$
\end{algorithmic}
\end{algorithm}

\subsection{Residual Coherence and the Optimal Clipping Bound}
\label{subsec:rc_proxy}
This subsection proves that residual coherence at the optimal clipping bound is near our chosen threshold $T$. The result (Theorem~\ref{thm:rc}) relies on Assumption~\ref{assm:tail_uncorrelated}, which precludes any clean linear relationship between how deep a vector lies in the tail and the direction that it points in. This can be viewed as a formalization of the rough idea that outliers tend to vary in different ways.
\begin{assumption}
\label{assm:tail_uncorrelated}
    Given MSE-optimal clipping bound $C^*$, a tail vector's excess magnitude is uncorrelated with its direction:
\begin{equation}
\label{eq:uncorrelated_assumption}
    \E{\calD}{(R_i-C^*)U_i \mid R_i > C^*} = \E{\calD}{R_i-C^* \mid R_i > C^*} \cdot \E{\calD}{U_i \mid R_i > C^*}.
\end{equation}
\end{assumption}
Assumption~\ref{assm:tail_uncorrelated} also rules out a pathological bad instance for $\rcalg$ in which, for example, half of the data consists of small vectors in a shared direction, the other half of the data consists of much larger vectors in a shared opposite direction, and $\rcalg$ stops prematurely because the unit vectors cancel at $C = \beta^0 = 1$. With Assumption~\ref{assm:tail_uncorrelated}, expected residual coherence is near $T = \sqrt{2d/\rhosum}$ at the optimal $C^*$.
\begin{theorem}
\label{thm:rc}
    Let $X=\{x_{1},\dots,x_{n}\} \subset \mathbb{R}^d$ be a dataset drawn i.i.d.\ from a distribution $\calD$ satisfying Assumption~\ref{assm:tail_uncorrelated}. Let $C^*$ denote the MSE-optimal bound, let $p_{C^*} = \P{\calD}{R_i > C^*}$, and define $\gamma$ by $\E{\calD}{R_i - C^* \mid R_i > C^*} = \gamma C^*$. Then $\E{\calD^n}{\rc{C^*}^2} = \frac{2d}{\gamma\rhosum}$, and
    \begin{equation*}
        \sqrt{\frac{n}{n-1}\left(\frac{2d}{\gamma\rhosum} - n p_{C^*}\right)} \le \E{\calD^n}{\rc{C^*}} \le \sqrt{\frac{2d}{\gamma\rhosum}}.
    \end{equation*}
\end{theorem}
\begin{proof}
    Define $Y_i = \mathbb{I}(R_i > C^*)U_i$ and $Z_i = \mathbb{I}(R_i > C^*)(R_i - C^*)U_i$, so that $\rc{C^*} = \left\|\sum_{i=1}^n Y_i\right\|_2$ and the bias derivative at $C^*$ is $\sum_{i=1}^n \sum_{j=1}^n \langle Z_i, Y_j \rangle$. By linearity of expectation over $[n]^2$,
    \begin{equation}
    \label{eq:thm_1}
        \E{\calD^n}{\sum_{R_i > C^*} \sum_{R_j > C^*} (R_i - C^*) \langle U_i, U_j \rangle} = \sum_{i=1}^n \sum_{j=1}^n \E{\calD^n}{\langle Z_i, Y_j \rangle}.
    \end{equation}
    For the diagonal terms ($i=j$), since $\langle U_i, U_i \rangle = 1$, the definition of $\gamma$ yields
    \begin{align*}
        \E{\calD}{\langle Z_i, Y_i \rangle} =&\  \E{\calD}{\mathbb{I}(R_i > C^*)(R_i - C^*)} \\
        =&\ p_{C^*} \E{\calD}{R_i - C^* \mid R_i > C^*} \\
        =&\ \gamma C^* p_{C^*} \\
        =&\ \gamma C^* \E{\calD}{\|Y_i\|_2^2}.
    \end{align*}
    For the off-diagonal terms ($i \neq j$), we apply Assumption~\ref{assm:tail_uncorrelated} and get
    \begin{align*}
        \E{\calD}{Z_i} &= p_{C^*}\E{\calD}{(R_i - C^*)U_i \mid R_i > C^*} \\
        &= p_{C^*}\E{\calD}{R_i - C^* \mid R_i > C^*} \cdot \E{\calD}{U_i \mid R_i > C^*} \\
        &= \gamma C^* p_{C^*} \E{\calD}{U_i \mid R_i > C^*} \\
        &= \gamma C^*\E{\calD}{Y_i}.
    \end{align*}
    Then by independence of $x_i$ and $x_j$ for $i \neq j$,
    \begin{align*}
        \E{\calD^2}{\langle Z_i, Y_j \rangle} &= \langle \E{\calD}{Z_i}, \E{\calD}{Y_j} \rangle \\
        &= \gamma C^* \langle \E{\calD}{Y_i}, \E{\calD}{Y_j} \rangle \\
        &= \gamma C^* \E{\calD^2}{\langle Y_i, Y_j \rangle}.
    \end{align*}
    We return to Equation~\ref{eq:thm_1} and get 
    \begin{equation}
        \sum_{i=1}^n \sum_{j=1}^n \E{\calD^n}{\langle Z_i, Y_j \rangle} = \gamma C^* \sum_{i=1}^n \sum_{j=1}^n \E{\calD^n}{\langle Y_i, Y_j \rangle} = \gamma C^* \cdot \E{\calD^n}{\rc{C^*}^2}. \label{eq:bias_derivative_rc}
    \end{equation}
    Setting the expected bias derivative to $\frac{2C^*d}{\rhosum}$ in Equation~\ref{eq:bias_derivative_rc} gives $\E{\calD^n}{\rc{C^*}^2} = \frac{2d}{\gamma\rhosum}$, and by Jensen's inequality
    \begin{equation*}
        \E{\calD^n}{\rc{C^*}} \le \sqrt{\E{\calD^n}{\rc{C^*}^2}} = \sqrt{\frac{2d}{\gamma\rhosum}}.
    \end{equation*}
    
    For the lower bound, let $\mu = \E{\calD}{Y_i}$. Since $Y_1, \dots, Y_n$ are i.i.d., we have $\E{\calD}{\|Y_i\|_2^2} = p_{C^*}$, $\E{\calD^2}{\langle Y_i, Y_j \rangle} = \langle \E{\calD}{Y_i}, \E{\calD}{Y_j} \rangle = \|\mu\|_2^2$ for all $i \neq j$, and $\left\|\E{\calD^n}{\sum_{i=1}^n Y_i}\right\|_2^2 = n^2\|\mu\|_2^2$. Expanding $\E{\calD^n}{\rc{C^*}^2}$ therefore gives
    \begin{align*}
        \E{\calD^n}{\rc{C^*}^2} &= \sum_{i=1}^n \E{\calD}{\|Y_i\|_2^2} + \sum_{i \neq j} \E{\calD^2}{\langle Y_i, Y_j \rangle} \\
        &= n p_{C^*} + n(n-1)\|\mu\|_2^2 \\
        &= n p_{C^*} + \frac{n-1}{n}\left\|\E{\calD^n}{\sum_{i=1}^n Y_i}\right\|_2^2.
    \end{align*}
    Rearranging, substituting in $\E{\calD^n}{\rc{C^*}^2} = \frac{2d}{\gamma\rhosum}$, and applying Jensen's inequality again yields
    \begin{align*}
        \E{\calD^n}{\rc{C^*}} &= \E{\calD^n}{\left\|\sum_{i=1}^n Y_i\right\|_2} \\
        &\geq \left\|\E{\calD^n}{\sum_{i=1}^n Y_i}\right\|_2 \\
        &= \sqrt{\frac{n}{n-1}\left(\frac{2d}{\gamma\rhosum} - n p_{C^*}\right)}.
    \end{align*}
\end{proof}

We note that Assumption~\ref{assm:tail_uncorrelated} is sufficient but perhaps not necessary. The next section demonstrates that $\rcalg$ performs well across a variety of datasets. No dataset was chosen to satisfy Assumption~\ref{assm:tail_uncorrelated}, and most only approximately satisfy it.

%% file: sections/experiments.tex
\section{Experiments}
\label{sec:experiments}
This section evaluates the empirical performance of $\rcalg$ on real-world datasets. For each dataset, we compare $\rcalg$ against baselines across a range of $\rhoclip$, keeping $\rhosum=0.1$ fixed. In each experiment, we compute a $\rhoclip$-CDP clipping bound $\tilde C$, clip the vectors to $\tilde C$, sum them, apply $\rhosum$-CDP spherical Gaussian noise scaled to $2\tilde C$, and finally compute the MSE $\frac{1}{d}\sum_{j=1}^d(\tilde y_j - y_j)^2$ between the noisy vector sum $\tilde y$ and the true (unclipped) vector sum $y$. $\rcalg$ uses the same candidate clip bounds $\{1.001^{j-1}\}_{j=1}^\infty$ throughout, and these bounds were taken from~\citet{D23} before running any experiments.\footnote{These bounds implicitly assume $C^* \geq 1$. If an end-user wishes to account for the possibility that $C^* \ll 1$, a simple solution is to additionally run $\rcalg$ on a sequence of shrinking $C_j$ until residual coherence \emph{exceeds} the threshold.} Complete experiment code may be found on Github~\cite{G26B}.

\subsection{Baselines}
\label{sec:baselines}
We evaluate $\rcalg$ against two baselines. The first is the $\uqalgshort$ algorithm introduced by~\citet{D23}, which uses $\atalg$ to privately estimate a target quantile of the vectors' $\ell_2$ norms. We use the same candidate clip bounds $\{1.001^{j-1}\}_{j=1}^\infty$ as $\rcalg$, again noting that the choice of 1.001 comes from~\citet{D23}. To evaluate the role that the quantile choice plays in its performance, we run it for each target quantile in $\{0.9, 0.95, 0.99\}$, noting that~\citet{D23} uses 0.99 as the single fixed choice in their experiments.\footnote{We also evaluated quantile 0.995, but it was never the best quantile choice, so we omit it.} 

The second baseline is a similar quantile-based algorithm due to~\citet{HLY21}, denoted $\hlyalg$. It privately estimates the $k^{th}$-highest $\ell_2$ norm using $k = \sqrt{\frac{2d}{\rhosum}}$. Their justification is similar to the MSE decomposition presented in Section~\ref{subsec:alg_desc}, with the simplifying (but loose and geometry-agnostic) application of $\|\sum_{R_i > C} (R_i-C)U_i\|_2 \leq \sum_{R_i > C}(R_i-C)$. Our implementation replaces their binary search-based quantile subroutine with Durfee's $\uqalgshort$ algorithm, as this consistently and significantly reduced empirical error. We also omit their pre-quantile shifting step for the same reason.\footnote{Their shifting step involves a division of the privacy budget across $d$. This may account for its relative success on their datasets ($d \ll 1000$) and failure on most of ours ($d \gg 1000$).}

Finally, we also evaluated $\abalg$~\cite{WZLDS+20}, $\bdkualg$~\cite{BDKU20}, and $\ksualg$~\cite{KSU20}. However, they never beat the above baselines and are omitted from our experiments. A longer discussion of these omitted methods appears in Appendix~\ref{subsec:omitted_methods}.

\subsection{Datasets}
\label{subsec:datasets}
Our experiments use ten real-world datasets (and, in the Appendix, an additional four light-tail datasets for completeness). For each dataset, we record the descriptive 4-tuple (dimensions $n \times d$, row, column, entry) below. Note that the first eight datasets only have nonnegative data, while the last two also have negative data. We set $\rhoclip$ appropriately according to Lemma~\ref{lem:rc_cdp}.
\begin{enumerate}
    \item Steam~\cite{T16}: (11350 x 3600, Steam user, game, \# hours played)
    \item Reddit~\cite{GGKSS+20}: (223388 x 153701, Reddit user, token, \# times token used in comments)
    \item LastFM~\cite{C10}: (10237 x 49602, LastFM user, music artist, \# plays), first 500,000 rows
    \item HN~\cite{N26}: (43679 x 18884, Hacker News user, story, comment or not), March 2026
    \item MovieLens~\cite{HK15}: (6022 x 3043, MovieLens user, movie, whether or not user rated movie)
    \item Pantry~\cite{NLM19}: (247659 x 10814, Pantry user, item, whether or not user rated item)
    \item Twitch~\cite{RMA21}: (100000 x 739991, Twitch user, stream, whether or not user watched stream)
    \item Foursquare~\cite{YZZY14}: (2293 x 61858, Foursquare Tokyo user, venue, \# visits)
    \item Jester~\cite{GRGP01}: (24983 x 100, Jester user, joke, rating in $[-10, 10]$)
    \item Slashdot~\cite{LHK10}: (44044 x 70284, Slashdot user, rated user, $\pm 1$ tag)
\end{enumerate}
We chose these datasets to represent a variety of contexts. Prior information consisted only of a dataset's brief qualitative description. Most datasets fit under the broad umbrella of ``internet data''. This is in part because we view such data as an important application for differential privacy, and in part because user-keyed internet data is most readily publicly available. Nonetheless, we suggest that the data contexts are heterogeneous enough to plausibly extend to other individual data settings.

As an aside, we note that a fully private histogram sometimes requires an initial private key selection stage to even generate the histogram domain. The domain in the Reddit dataset, for example, is a function of user data, as a user can introduce a token. We view private key selection as orthogonal to the private sum problem considered here and omit it from our pipeline, essentially treating the column headers of each dataset as public information. A closer investigation of the relationship between key selection and clipping bound selection may be an interesting direction for future work.

Finally, for completeness we also evaluated the methods on four datasets whose qualitative descriptions suggested light tails. As described in the introduction, most methods generally did well on these datasets, so we defer the plots to Appendix~\ref{subsec:light_tails}. Here, we simply note that our conceptual tailoring of $\rcalg$ to heavy-tailed data does not lead to poor performance on light-tailed data. Instead, $\rcalg$ matches the best (low) error of the baselines.

\subsection{Results}
\subsubsection{Error}
\label{subsubsec:error}
MSE results by dataset, method, and $\rhoclip$ value appear in Figure~\ref{fig:plots} and are summarized in Table~\ref{tab:overall}. The primary metric is MSE ratio with respect to $\optalg$, which non-privately computes the MSE-optimal $C^*$. $\rcalg$ achieves the strongest overall empirical performance. Its mean MSE ratio at $\rhoclip = 0.05$\footnote{Relative performance does not vary much across $\rhoclip$ (see Figure~\ref{fig:plots}). $\rhoclip = 0.05$ is a representative point.} is best on 6/10 datasets, second-best on the remaining 4/10, and suboptimal by at most a factor of 1.32 across datasets. In contrast, every other baseline is suboptimal by at least a factor of 2 (and often much more) on at least half the datasets.

\begin{table}[h]
    \centering
    \small
    \setlength{\tabcolsep}{0pt}
    \begin{tabular*}{\textwidth}{@{\extracolsep{\fill}}lcccccccccc@{}}
        \toprule
        \textbf{Dataset} & Steam & Reddit & LastFM & HN & MovieLens & Pantry & Twitch & Foursquare & Jester & Slashdot \\
        \midrule
        $\rcalg$ & 1.02 & 1.03 & \textbf{1.32} & 1.04 & 1.12 & \textbf{1.05} & \textbf{1.22} & \textbf{1.10} & \textbf{1.05} & \textbf{1.04} \\ 
        $\hlyalg$ & 1.23 & 1.08 & 5.45 & 3.97 & 1.43 & 2.94 & 8.40 & 4.31 & 1.27 & 6.93 \\
        $\uqalgshort$, 0.9 & 1.64 & 6.26 & 5.31 & \textbf{1.01} & \textbf{1.05} & 1.58 & 5.67 & 35.65 & 3.24 & 2.08 \\
        $\uqalgshort$, 0.95 & \textbf{1.01} & 2.76 & 11.15 & 1.41 & 1.32 & 1.06 & 7.74 & 67.95 & 1.22 & 3.96 \\
        $\uqalgshort$, 0.99 & 2.33 & \textbf{1.01} & 41.40 & 5.06 & 2.06 & 1.41 & 11.78 & 136.66 & 1.10 & 13.87 \\
        \midrule
        $\mathbf{\hat F(C^*)}$ & 0.96 & 0.99 & 0.51 & 0.89 & 0.86 & 0.97 & 0.08 & $\approx$ 0.00 & 0.98 & 0.65 \\ 
        \bottomrule
    \end{tabular*}
    \vspace{3pt}
    \caption{First five rows: mean error ratio (relative to optimal clipping) across 100 trials at $\rhoclip = 0.05$ for each method (lower is better), with the lowest error ratio bolded for each dataset. Bottom row: empirical CDF at the optimal clipping bound $C^*$ across datasets.}
    \label{tab:overall}
\end{table}

Figure~\ref{fig:plots} shows that some instances of $\uqalgshort$ can achieve lower average error than $\rcalg$ on some datasets. However, the empirical CDF at the optimal clipping bound $C^*$ varies widely across datasets, from $<0.01$ to $0.99$ (see last row of Table~\ref{tab:overall}). Consequently, even a non-private fixed quantile cannot perform well in general. In contrast, $\rcalg$ incorporates the underlying data geometry to achieve low error across a range of datasets.

\begin{figure}
  \centering
  \includegraphics[scale=0.46]{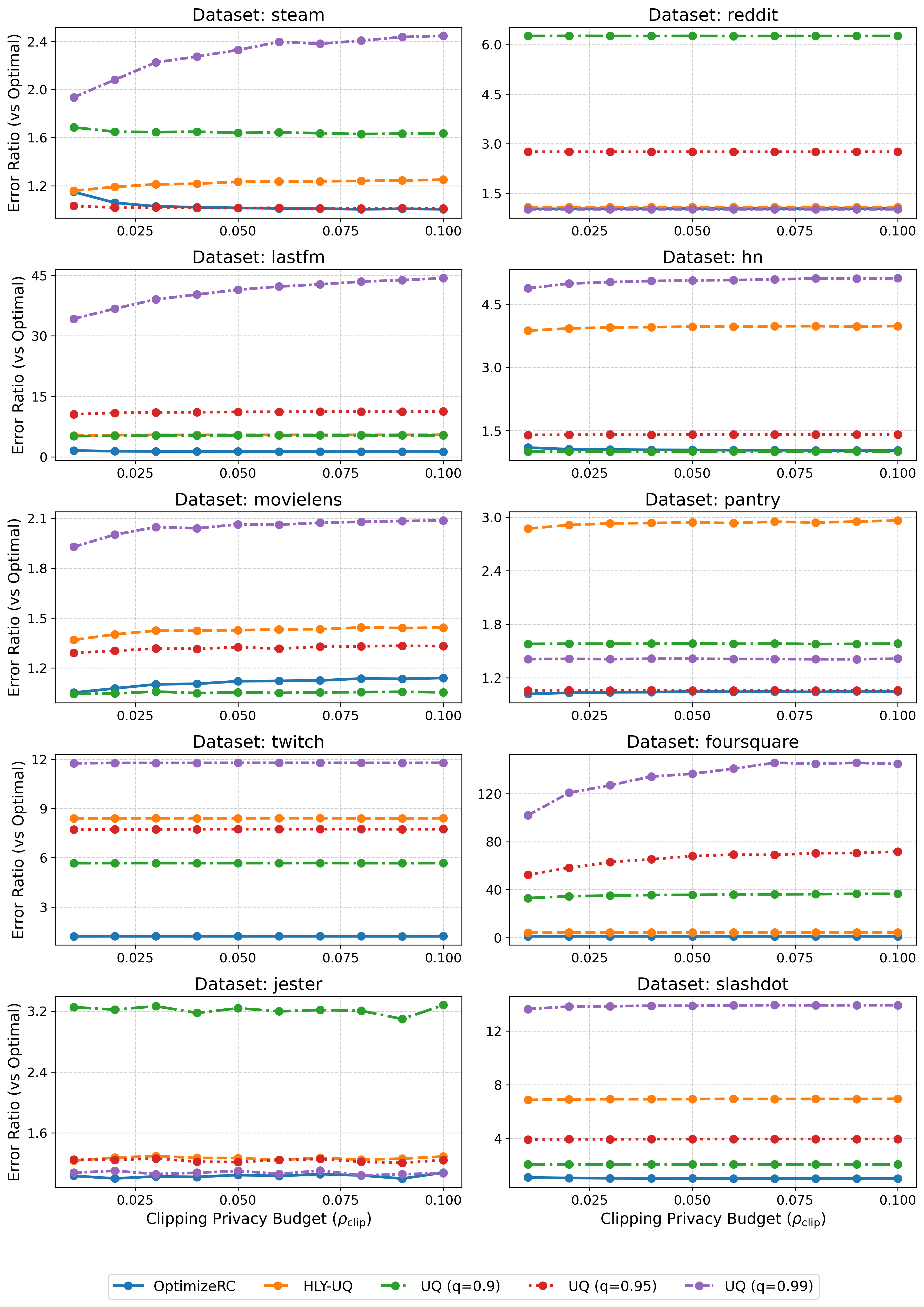}
  \caption{Plots of $\rcalg$'s performance against $\uqalgshort$ ($q = \{0.9, 0.95, 0.99\}$) and $\hlyalg$. Each line is a different method, and each point on the line represents the mean of 100 trials per $\rhoclip \in \{0.01, 0.02, \ldots, 0.1\}$. We plot ratios of the mean MSE with respect to the (non-private) MSE-optimal $\optalg$, so a lower value is better.}
  \label{fig:plots}
\end{figure}

Additional plots in Appendix~\ref{subsec:more_plots} (Figure~\ref{fig:rc_plots}) record the relationship between the residual coherence of successive clipping bounds, the threshold $\sqrt{2d / \rhosum}$, and the optimal clipping bound $C^*$. We find that residual coherence generally intersects the threshold relatively close to $C^*$, and even when it errs, it does not err as catastrophically as the quantile methods. Perhaps surprisingly, we also find that residual coherence is monotonically decreasing in $C$ even on general data\footnote{It is always monotonically decreasing on nonnegative data.}, as the chosen datasets lack the strongly opposing vectors that would be required to make residual coherence non-monotonic.

Plots of $\tilde C / C^*$ for different algorithms also appear in Appendix~\ref{subsec:more_plots} (Figure~\ref{fig:clip_plots}). We find that $\hlyalg$ consistently overestimates the optimal $C^*$, and there is no (heavy-tailed dataset, $\rhoclip$) pair out of 100 on which $\hlyalg$ obtains lower error than $\rcalg$. This follows from $\hlyalg$'s simplifying bound $\|\sum_{R_i > C} (R_i-C)U_i\|_2 \leq \sum_{R_i > C}(R_i-C)$ (see Section~\ref{sec:baselines}), which consistently leads to overestimates relative to $\rcalg$.

\subsubsection{Testing Assumption 1}
\label{subsubsec:testing_assumption}
Our utility result (Theorem~\ref{thm:rc}) assumes that the underlying data distribution satisfies Assumption~\ref{assm:tail_uncorrelated}. The following experiment tests this assumption by evaluating the normalized $\ell_2$ distance between the LHS and RHS of Equation~\ref{eq:uncorrelated_assumption} at $C^*$ across datasets.

\begin{table}[h]
    \centering
    \small
    \setlength{\tabcolsep}{0pt}
    \begin{tabular*}{\textwidth}{@{\extracolsep{\fill}}lcccccccccc@{}}
        \toprule
        \textbf{Dataset} & Steam & Reddit & LastFM & HN & MovieLens & Pantry & Twitch & Foursquare & Jester & Slashdot \\
        \midrule
        $\frac{\|\mathbf{L} - \mathbf{R}\|_2}{\|\mathbf{L}\|_2}$ & 0.10 & 0.04 & 0.27 & 0.27 & 0.13 & 0.29 & 0.27 & 0.18 & 0.14 & 0.38 \\
        \bottomrule
    \end{tabular*}
    \vspace{3pt}
    \caption{Normalized $\ell_2$ distance between $\mathbf{L} = \hat{\mathbb{E}}[(R_i - C^*)U_i \mid R_i > C^*]$ and $\mathbf{R} = \hat{\mathbb{E}}[R_i - C^* \mid R_i > C^*]\hat{\mathbb{E}}[U_i \mid R_i > C^*]$ (LHS and RHS in Assumption~\ref{assm:tail_uncorrelated}). Values far from 0 correspond to worse assumption violations.}
    \label{tab:assumption}
\end{table}

As shown in Table~\ref{tab:assumption}, Assumption~\ref{assm:tail_uncorrelated} is only approximately satisfied on most of our datasets, suggesting that Theorem~\ref{thm:rc} is conservative: weaker properties may govern $\rcalg$'s practical performance.

\subsubsection{Time}
\label{subsubsec:time}
Runtimes for evaluations on a typical personal computer appear in Table~\ref{tab:time}, where the pre-computation for a given $\rhoclip$ is amortized over 10 trials of $\atalg$. $\rcalg$ is typically slower than $\uqalgshort$ and $\hlyalg$ by approximately an order of magnitude, but it is still fast in absolute terms, consistently running in $\leq$ 1s even when the number and length of vectors are in the hundreds of thousands. This speedup is a consequence of the fast algorithm described in Section~\ref{subsec:rc_imp} and its implementation using SciPy's sparse matrices. We suggest that the mild increase in absolute runtime is an acceptable trade for $\rcalg$'s strong error improvements.

\begin{table}[h]
    \centering
    \small
    \setlength{\tabcolsep}{0pt}
    \begin{tabular*}{\textwidth}{@{\extracolsep{\fill}}lcccccccccc@{}}
        \toprule
        \textbf{Algorithm} & Steam & Reddit & LastFM & HN & MovieLens & Pantry & Twitch & Foursquare & Jester & Slashdot \\
        \midrule
        $\rcalg$ & 0.14 & 0.7 & 0.4 & 0.1 & 0.06 & 0.05 & 1 & 0.2 & 0.07 & 0.13 \\
        $\uqalgshort$, 0.9 & 0.01 & 0.03 & 0.02 & 0.007 & 0.006 & 0.02 & 0.02 & 0.008 & 0.01 & 0.007 \\
        $\uqalgshort$, 0.95 & 0.01 & 0.03 & 0.02 & 0.007 & 0.006 & 0.02 & 0.02 & 0.009 & 0.01 & 0.007 \\
        $\uqalgshort$, 0.99 & 0.02 & 0.03 & 0.02 & 0.009 & 0.006 & 0.02 & 0.02 & 0.01 & 0.01 & 0.009 \\
        $\hlyalg$ & 0.01 & 0.03 & 0.02 & 0.009 & 0.006 & 0.02 & 0.02 & 0.006 & 0.01 & 0.008 \\
        \bottomrule
    \end{tabular*}
    \vspace{3pt}
    \caption{Average runtime (in seconds) of bound estimation algorithms.}
    \label{tab:time}
\end{table}

%% file: sections/discussion.tex
\section{Discussion}
\label{sec:discussion}
This paper proposes a method, $\rcalg$,  for privately choosing a clipping bound for summing $d$-dimensional vectors. Incorporating the dataset's geometry makes it possible to privately select clipping bounds that yield lower MSE than those produced by the comparable baselines in the literature. We note that the small number of algorithmic choices required from the operator, such as the candidate clipping bound schedule $C_1, C_2, \ldots$, are fixed to the same untuned quantities across our experiments. We present this as evidence that it is empirically possible to privately choose good sum clipping bounds with no input from the algorithm operator.

We close with a few ways $\rcalg$ might be improved. First, while Theorem~\ref{thm:rc} establishes sufficient conditions under which residual coherence intersects the desired threshold at the optimal $C^*$, it does not provide necessary conditions (as shown in Section~\ref{subsubsec:testing_assumption}) or prove that the output $\tilde C$ is near-optimal. Second, to the best of our knowledge, little is known about residual coherence. The material in Section~\ref{subsec:rc_proxy} presents some evidence that it is a plausible low-sensitivity proxy for the bias derivative, but better identification of the data regimes in which it performs well -- or better low-sensitivity alternatives -- may be an interesting avenue for future work. Finally, we have not attempted to apply $\rcalg$ to one of the main venues for private sums: gradient clipping in  DP-SGD. Briefly, this is because improving DP-SGD does not cleanly reduce to minimizing the MSEs of gradient sums. As (non-DP) SGD demonstrates, clipping can improve convergence stability, and the appropriate degree of clipping is not well understood even absent privacy. Nonetheless, exploring $\rcalg$ or similar clipping bound selection approaches in DP-SGD may be useful.

%% file: sections/ai.tex
\section{AI Use Disclosure}
We used various versions of Gemini for experiment code and minor math assistance in the proof of Theorem~\ref{thm:rc}. All prose is human-written and the logic of all AI code and math have been human-checked. We take full responsibility for the correctness of this work.

%% file: sections/appendix.tex
\section{Omitted Methods}
\label{subsec:omitted_methods}
This section discusses three omitted alternative methods for private clipping bound selection: $\abalg$~\cite{WZLDS+20}, $\bdkualg$~\cite{BDKU20}, and $\ksualg$~\cite{KSU20}. At a high level, the three methods pursue different algorithmic strategies, but they all include a step that deliberately chooses a clipping bound large enough to leave most points unclipped. This introduces the same problems as any method that privately estimates a large fixed quantile (see Section~\ref{subsubsec:error}) and is in each case compounded by other algorithmic details. Explanations for each method follow.

$\abalg$ computes, e.g., private sums by constructing a private histogram over the data's $\ell_2$ norms, computes a threshold count $T$ unlikely to be crossed by any truly empty bin, taking the largest bin with count at least $T$ as the clipping bound, and computing a standard clipped private sum. In addition to the aforementioned problems in taking a fixed quantile as the clipping bound, this histogram approach is additionally a less accurate private quantile primitive than $\uqalg$'s $\atalg$, so it is empirically dominated by $\uqalg$.

$\bdkualg$ computes a private mean by projecting the data onto some ball containing the true mean, estimating a simple private clipped mean using additive noise, and repeating this process with geometrically shrinking balls. The clipping bound chosen for each clipped mean is set using Gaussian tail bounds calibrated to probably contain all data points. This does not yield an obvious analogue when the data is heavy-tailed, and the general clipping strategy is effectively attempting to use the maximum norm as a clipping bound. Finally, we note that the experiments of~\citet{HLY21} generally show $\hlyalg$ outperforming $\bdkualg$ outside of well-specified Gaussian data.

$\ksualg$ is designed for heavy-tailed private mean estimation. It constructs a private histogram over each dimension to obtain an initial coordinate-wise mean estimate, truncates data to a ball around this estimate, and finally computes a standard private clipped mean  using the ball. However, the initial mean estimate requires an input interval $[-R,R]$ containing the true mean and a bound on the data distribution's $k^{th}$ moment and further requires $n \gg d$ samples for the result to be accurate. It is not clear how to set either parameter in general, and only 4/10 of our primary datasets have $n \geq 2d$. Furthermore, the second estimation step chooses a truncated ball radius based on worst-case concentration bounds for distributions with bounded $k^{th}$ moment, calibrated to contain nearly all data points. This again inherits the aforementioned problems of clipping to accommodate the largest data points.

Finally, we note that evaluations of the above methods on our datasets never beat the three methods discussed in the text, though the number of user-specified parameters in $\bdkualg$ and $\ksualg$ make a standardized untuned comparison difficult.

\newpage 
\section{Additional Heavy-Tail Plots}
\label{subsec:more_plots}
\enlargethispage{1in}

\begin{figure}[!ht]
  \centering
  \includegraphics[scale=0.47]{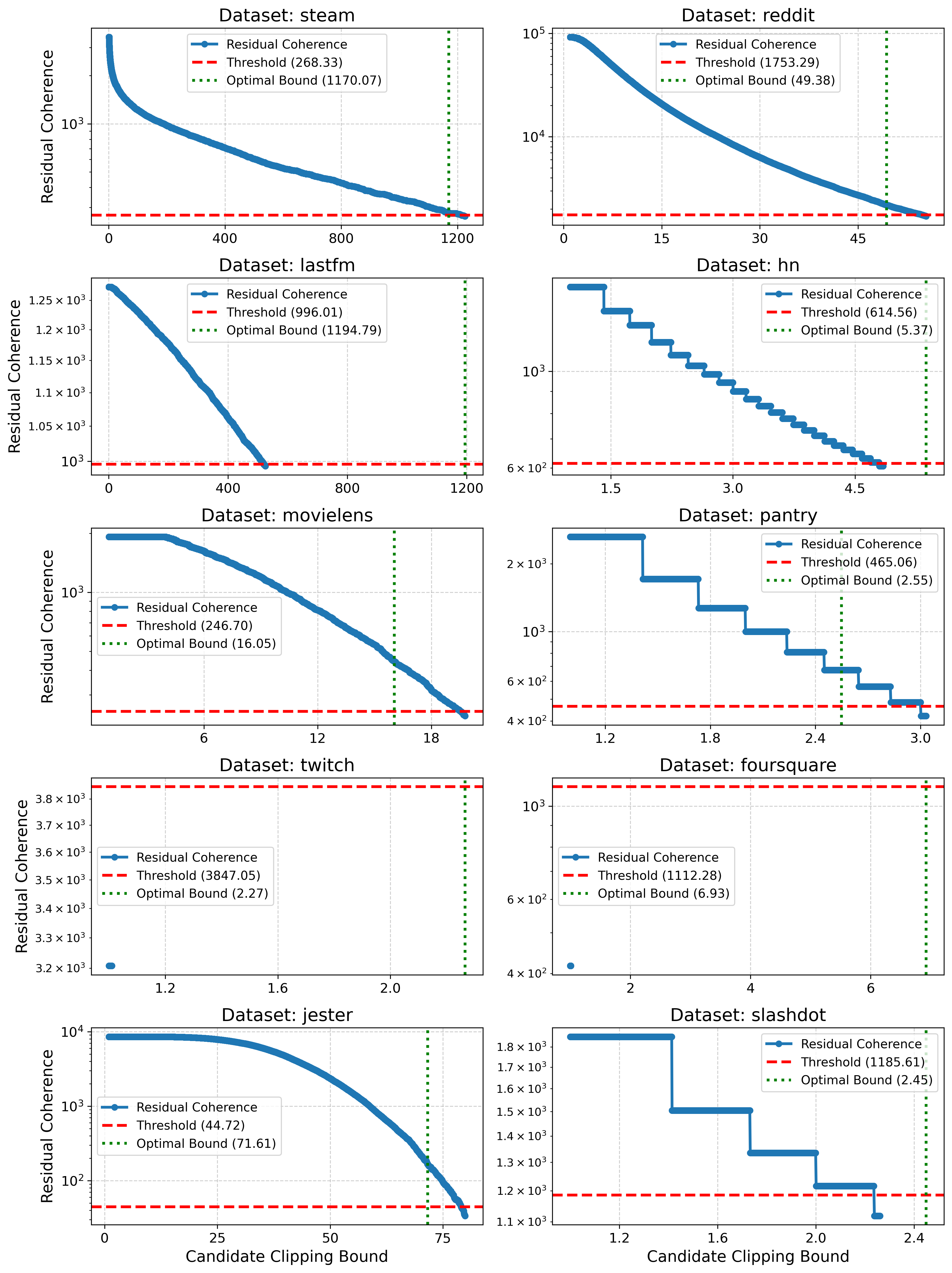}
  \caption{\small Plots of residual coherence vs. clipping bounds across all evaluated datasets. The horizontal line is $-T$ in $\rcalg$. Note that HN, Pantry, and Slashdot have a distinctive ``staircase'' shape because the small clipping bound range accentuates the underlying integer values. For the Twitch and Foursquare datasets, residual coherence is immediately below the threshold, and $\rcalg$ outputs 1. This aggressively clipped bound nonetheless beats the $\uqalg$ methods, which select bounds that are too large (see Figure~\ref{fig:clip_plots}).}
  \label{fig:rc_plots}
\end{figure}

\begin{figure}[h]
  \centering
  \includegraphics[scale=0.47]{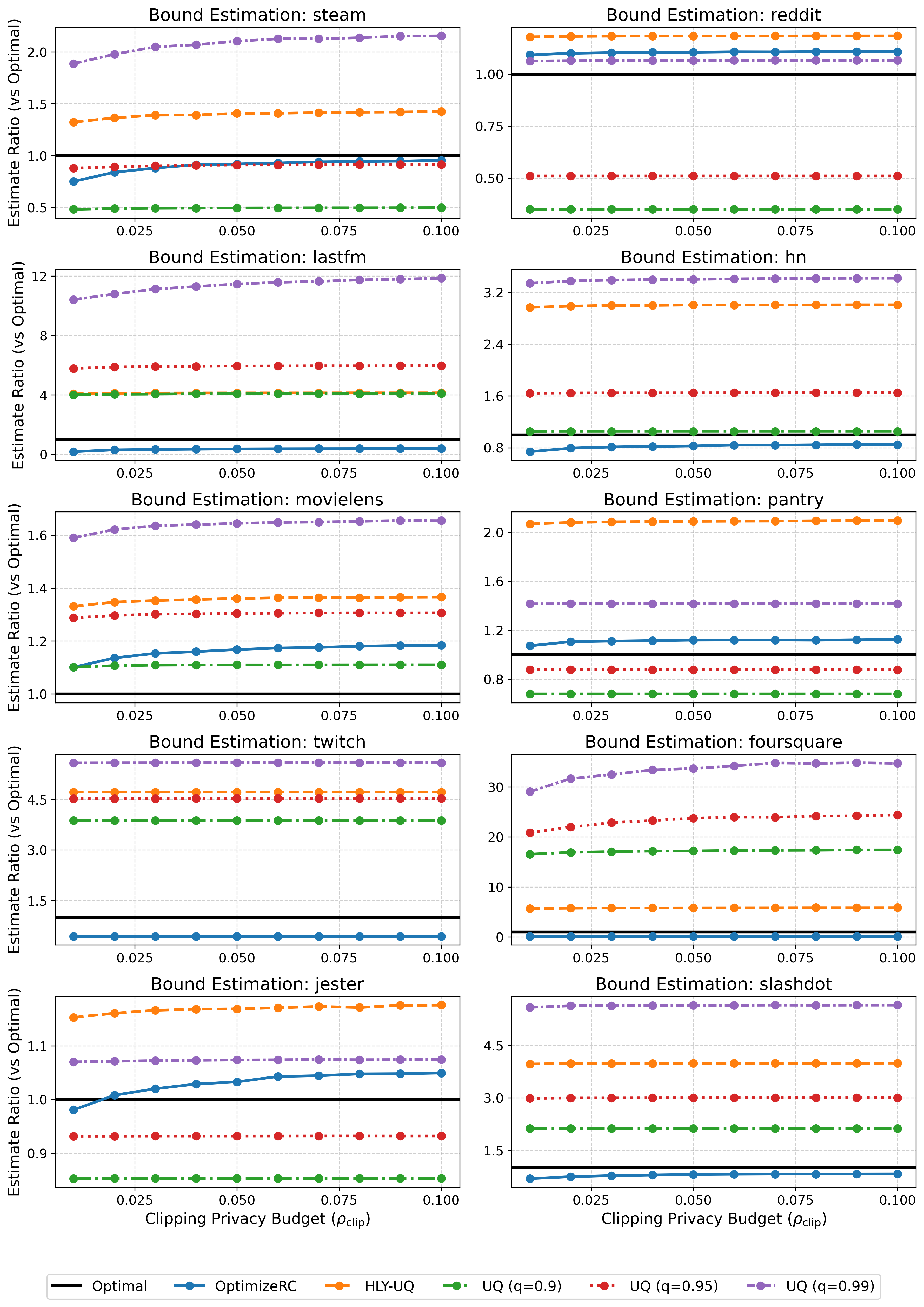}
  \caption{Plots illustrating the ratio between the estimated clipping bound and the optimal $C^*$. A value $<1$ means the estimated bound is smaller, and a value $>1$ means it is larger.}
  \label{fig:clip_plots}
\end{figure}

\clearpage

\section{Light-Tailed Experiments}
\label{subsec:light_tails}
The light-tailed experiments use the following four datasets. We again record descriptive 4-tuples (dimensions $n \times d$, row, column, entry) below.
\begin{enumerate}
     \item ANSUR~\cite{GCCBM+88}: (1774 x 132, individual, measurement type, size in cm)
    \item ATUS~\cite{A16}: (5008 x 71, individual, ATUS activity, daily time spent in hours), restricted to first 100,000 rows
    \item NFL~\cite{S23}: (2555 x 6, NFL combine participant, drill, score)
    \item Powerlifting~\cite{O24}: (45268 x 5, lifter, metric, value), restricted to first 100,000 rows
\end{enumerate}
\begin{figure}[h]
  \centering
  \includegraphics[scale=0.4]{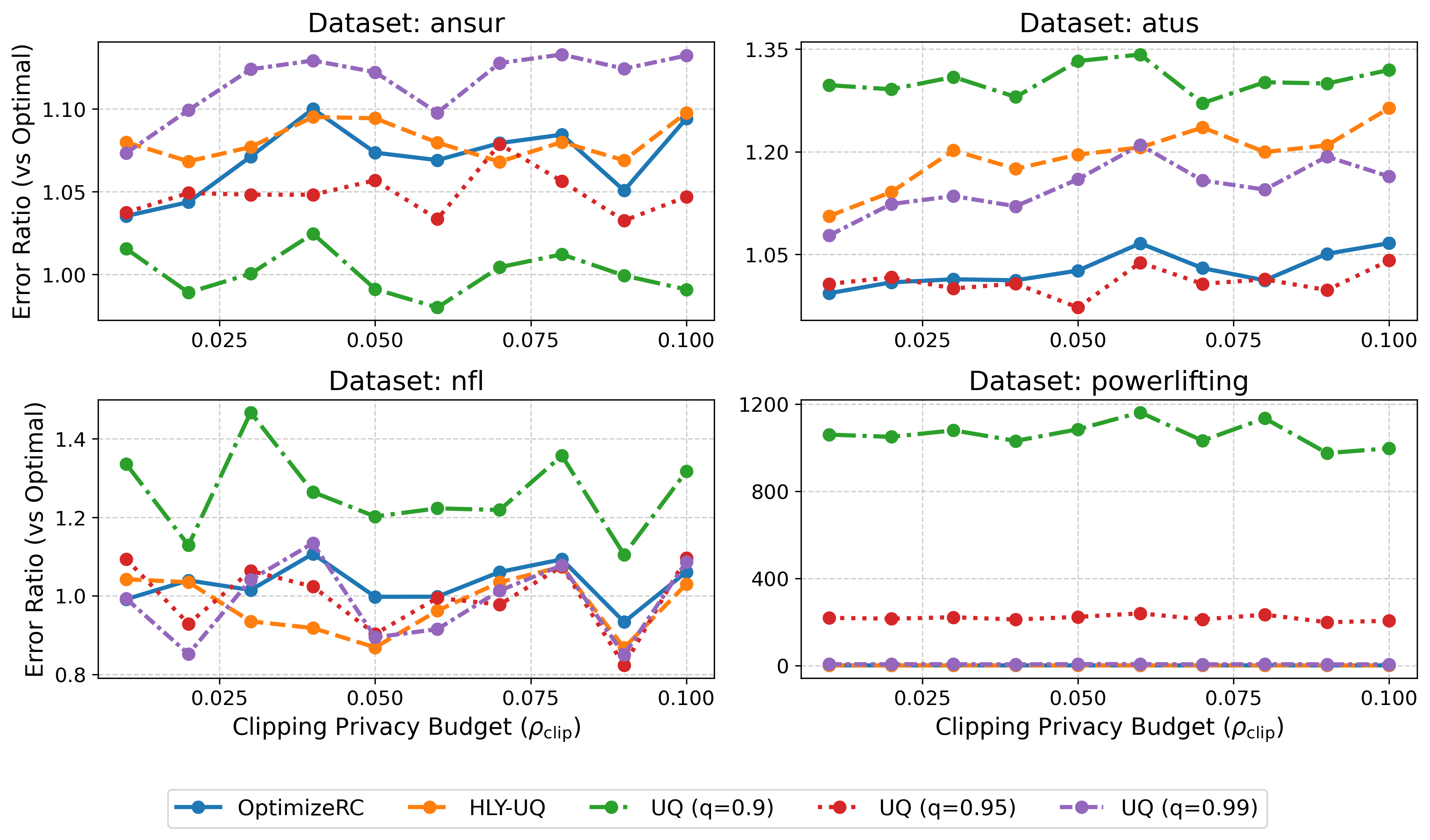}
  \caption{An analogue of Figure~\ref{fig:plots}, but for light-tailed datasets. Note that methods occasionally beat $\optalg$ on NFL because the low $d$ increases the noise variance.}
  \label{fig:light_plots}
\end{figure}

Note that $\rcalg$'s utility is similar to that of the heavy-tail setting, though the $\uqalg$ 0.99 and $\hlyalg$ baselines now often perform similarly well. Finally, we include clipping bound plots for completeness.

\begin{figure}[h]
  \centering
  \includegraphics[scale=0.4]{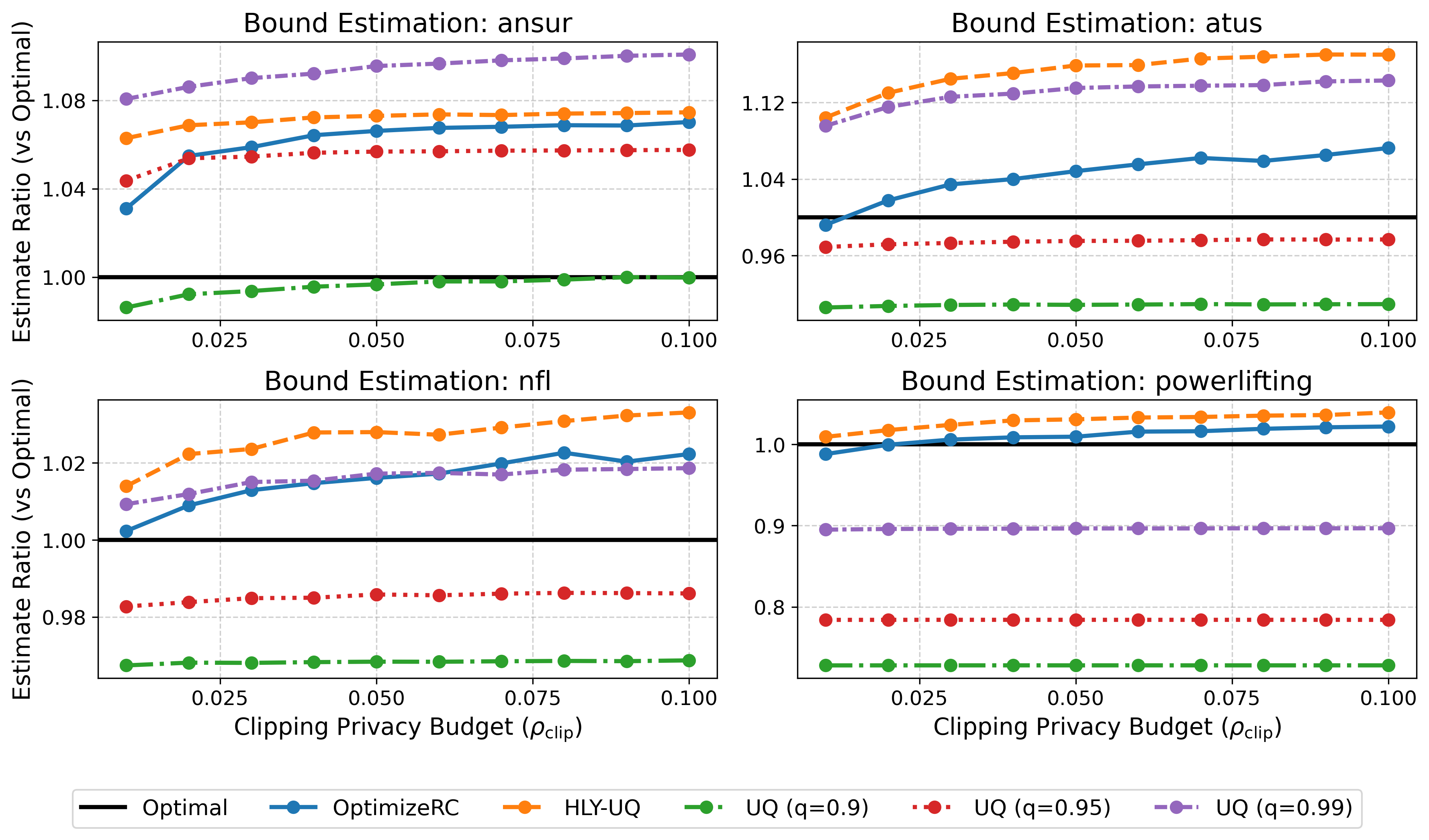}
  \caption{An analogue of Figure~\ref{fig:clip_plots}, but for light-tailed datasets.}
  \label{fig:light_clip_plots}
\end{figure}